\documentclass{article}

\usepackage{amsmath}
\usepackage{amssymb}
\usepackage{mathtools}
\usepackage{color}
\usepackage{graphicx}
\usepackage{framed}
\usepackage{multirow}
\usepackage[hyphens]{url}
\usepackage[compress,authoryear]{natbib}
\usepackage{wrapfig}
\usepackage{rotating}
\usepackage{csquotes}
\usepackage{listings}
\usepackage[pscoord]{eso-pic}
\usepackage{bbold}
\usepackage{amsthm}

\newtheorem{theorem}{Theorem}

\title{The Bow-Tie Centrality: \it A Novel Measure for Directed
  and Weighted Networks with an Intrinsic Node Property}

\author{James B. Glattfelder \\
{\small \it  Department of Banking and Finance, University of Zurich \vspace{-0.2cm}}\\ 
{\small \it   Andreasstrasse 15, 8015 Zurich, Switzerland \vspace{-0.2cm}}\\ 
{\small \it  james.glattfelder@uzh.ch \vspace{-0.2cm}}\\ 
}

\date{\small \today}

\begin{document}

\maketitle

\begin{abstract}
  Today, there exist many centrality measures for assessing the
  importance of nodes in a network as a function of their position and
  the underlying topology. One class of such measures builds on
  eigenvector centrality, where the importance of a node is derived
  from the importance of its neighboring nodes. For directed and
  weighted complex networks, where the nodes can carry some intrinsic
  property value, there have been centrality measures proposed that
  are variants of eigenvector centrality. However, these expressions
  all suffer from shortcomings. Here, an extension of such centrality
  measures is presented that remedies all previously encountered
  issues. While similar improved centrality measures have been
  proposed as algorithmic recipes, the novel quantity that is
  presented here is a purely analytical expression, only utilizing the
  adjacency matrix and the vector of node values. The derivation of
  the new centrality measure is motivated in detail. Specifically, the
  centrality itself is ideal for the analysis of directed and weighted
  networks (with node properties) displaying a bow-tie topology. The
  novel bow-tie centrality is then computed for a unique and extensive
  real-world data set, coming from economics. It is shown how the
  bow-tie centrality assesses the relevance of nodes similarly to
  other eigenvector centrality measures, while not being plagued by
  their drawbacks in the presence of cycles in the network.
\end{abstract}



\section{Introduction}

Centrality measures have a long history in the social sciences as a
structural attribute of nodes in a network
\cite{katz1953new,hubbell1965input,bonacich1972factoring,freeman1978centrality,bonacich1987power}.
The intuition is to identify important nodes depending on their
network position. To this day, centrality measures remain a
fundamental concept in network analysis
\cite{bonacich2001eigenvector,borgatti2006graph,newman2006structure} and find their
application in networks form physics and biology
\cite{freeman2008going} to economics
\cite{schweitzer2009economicsci,glattfelder2013decoding,glattfelder2019architecture}.

The centrality measure that is proposed here utilizes the maximal
information available from a network. For one, the links are expected
to be directed and can have weights. Moreover, it is assumed that the
nodes have an intrinsic degree of freedom. In detail, this is a
non-topological state variable describing a property value of the
nodes. An example of such a network is an ownership network, where the
links represent weighted and directed shareholding relations and some
of the nodes, representing firms, are assigned an economic value
\cite{glattfelder2009backbone,vitali2011network,glattfelder2013decoding,glattfelder2019architecture}.

In any directed network, each connected component (CC) can display a
bow-tie topology. This happens once a strongly connected component
(SCC) forms in the CC. In a SCC, each node is connected to each other
node in the SCC via a direct or indirect path. While a CC can contain
SCCs of various sizes, most real-world complex networks have a
dominant SCC, called a core. Once the core is defined, the bow-tie
naturally forms around it, as seen in Fig. \ref{fig:bt}, with the
various bow-tie components. Some real-world examples of complex
networks with a bow-tie topology are:
\begin{itemize}
\item ownership networks
  \cite{glattfelder2009backbone,vitali2011network,glattfelder2013decoding,glattfelder2019architecture}
\item the World Wide Web \cite{broder00,donato2008mining}
\item production networks \cite{fujiwara2008large}
\item online video social networks \cite{benevenuto2009video} 
\item Java online expertise forums  \cite{zhang2007expertise}
\end{itemize}

\begin{figure}[t]
\includegraphics[width=0.5\textwidth]{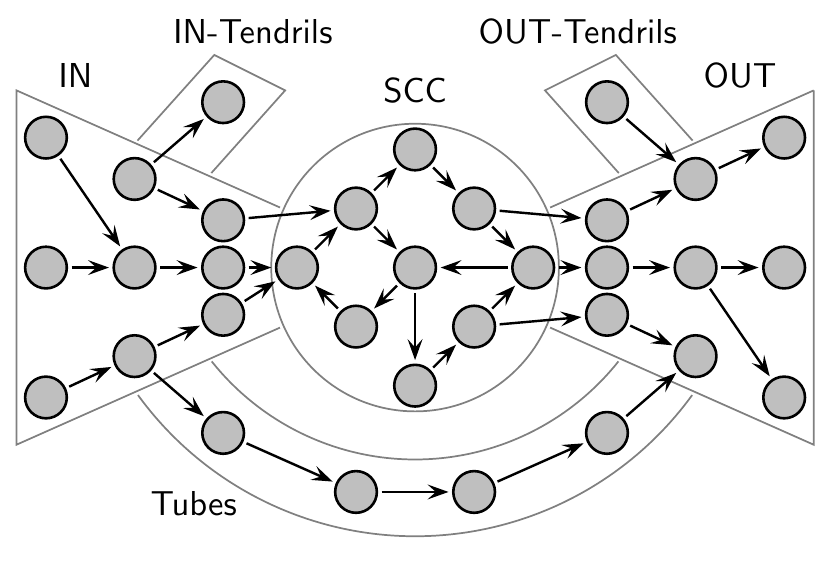}
\vspace*{8pt}
\caption{Bow-tie network topology.  An example consisting
  of an in-section (IN), an out-section (OUT), a strongly connected
  component (SCC) or core, and tubes and tendrils (TT).}
\label{fig:bt}
\end{figure}

\section{The Evolution of Centrality}

\subsection{Eigenvector Centrality}

A lot of attention has been devoted to feedback-type centrality
measures. These are based on the idea that a node is more central the
more central its neighboring nodes themselves are. This notion leads
to a set of equations which need to be solved simultaneously. In
general, this type of centrality is also categorized as eigenvector
centrality.  The colloquialism ``the importance of a node depends on
the importance of the neighboring nodes'' can be generically
quantified as
\begin{equation}
\label{eq:motto}
c_i = \sum_j A_{ij} c_j,
\end{equation}
where $A$ is the adjacency matrix of the graph and $c_i$ denotes the
centrality score of node $i$. More generally, by understanding $c$ as
an eigenvector of the adjacency matrix, the last equation can be
reformulated as
\begin{equation}
\label{eq:mottoev}
\lambda c = A c,
\end{equation}
with the eigenvalue $\lambda$. Note that Google's search engine is
based on a variant of eigenvector centrality, called PageRank
\cite{brin1998anatomy,page1999pagerank}.

Eigenvector-based centralities have been extensively studied in the
literature. For instance, \cite{bonacich2001eigenvector} introduced the following
variation
\begin{equation}
\label{eq:alpha}
\lambda c = \alpha A c + e,
\end{equation}
where $\alpha$ is a parameter and the vector $e$ represents an exogenous  source. If $e$
is assumed to be a vector of ones, the solution of Eq. (\ref{eq:alpha}) can be 
related to the well-known centrality measure introduced in \cite{katz1953new}.
Transitioning to weighted and directed graphs, a further refinement is given by the 
Hubbell index $c^H$  \cite{hubbell1965input}. Similarly to the term $e$ above, now the 
nodes are thought to posses an intrinsic importance $c^{0}$, to which
the importance from being connected to neighboring nodes is added.
In kinship to Eq. (\ref{eq:alpha}), the new centrality measure is defined as
\begin{equation}
\label{eq:hubbell-centrality}
c^H = W c^H +c^{0},
\end{equation}
where $W$ is the weighted adjacency matrix of the directed network. 
The solution is given by
\begin{equation}
\label{eq:hubbell-centralitysol}
c^H = (\mathbb{1}-W)^{-1} c^{0}.
\end{equation}
For the matrix $(\mathbb{1}-W)$ to be non-negative and non-singular, a
sufficient condition is that the Perron-Frobenius root is smaller than one,
\mbox{$\lambda(W)<1$}. This is ensured by the requirement that in each
strongly connected component $\mathcal{S}$ there exists at least one
node $j$ such that $\sum_{i\in \mathcal{S}} W_{ij}<1$
\cite{glattfelder2009backbone}. A similar centrality measure is found
in \cite{bonacich2001eigenvector}.

A final variant of eigenvector centrality, setting the stage for the new 
measure to be introduced in the following,  was defined in
\cite{bonacich1987power}
\begin{equation}
\label{eq:bondef}
c_i(\alpha,\beta) = \sum_j (\alpha +\beta c_j) A_{ij},
\end{equation}
with the solution
\begin{equation}
\label{eq:ac}
c(\alpha,\beta) = \alpha (\mathbb{1} - \beta A)^{-1} A e,
\end{equation}
and $e$ being the column vector of ones and $\alpha, \beta$ the
parameters to be chosen. This new centrality measure is essentially a
refinement of \cite{katz1953new} and \cite{hubbell1965input}. While it
was originally defined in terms of the adjacency matrix $A$, it can be
recast in the context of weighted and directed networks utilizing $W$.
Note that $W_{ij} \in [0,1]$ and $\sum_j W_{ij} \leq 1$ (i.e.,
column-stochastic).

Additional information on eigenvector-based centrality variants and general
node importance in complex networks can be found in \cite{martin2014localization}
and \cite{lu2016vital}, respectively.

\subsection{Centrality in Ownership Networks}

Prompted by the study of firms connected through a network of
cross-shareholdings \cite{brioschi1989risk,briosch1995equity}
proposed an algebraic model, based on the input-output matrix
methodology introduced to economics in \cite{leontief1986input}, to
calculate the value of the firms. This methodology can be generalized
to ownership networks and recast in the context of centrality
\cite{vitali2011network,glattfelder2013decoding}. The corresponding
equation defining the centrality $\chi$ is found to be
\begin{equation}
\label{eq:rec}
 \chi = W \chi + W v,
\end{equation}
where $v$ is a vector representing the economic values of the nodes.
Note, however, that $v$ can be a general non-topological node
property, allowing the methodology to be applied to generic networks.
Eq.  (\ref{eq:rec}) can be interpreted as follows: A node's centrality
score is given by the centrality scores of its neighbors plus the
neighbors intrinsic properties. The solution is given by
\begin{equation}
\label{eq:chi}
\chi = (\mathbb{1}- W)^{-1} W v =: \widetilde W  v.
\end{equation}
In other words, the centrality $\chi$ can be derived from the
centrality $c(\alpha,\beta)$ from Eq. (\ref{eq:ac}), by setting
$\alpha=\beta=1$ and replacing the vector $e$ with the node properties
defined by $v$.

Note that by using the series expansion 
\begin{equation}
\label{eq:iminaassum}
   (\mathbb{1}- W)^{-1} = \mathbb{1} + W + W^2 + W^3 +\dots
\end{equation}
one finds that
\begin{equation}
\label{eq:wtilddef}
\widetilde W =  (\mathbb{1}- W)^{-1} W = W  (\mathbb{1}- W)^{-1} = \sum_{n=1}^\infty W^n,
\end{equation}
and the centrality matrix equation is
\begin{equation}
\label{eq:matrixwtilddef}
  \widetilde W = W + W \widetilde W = W + \widetilde W W. 
\end{equation}

The centrality $\chi$ has a direct economic interpretation in terms of
the value of the total portfolio of shareholders
\cite{glattfelder2019architecture}. The direct portfolio 
value is the aggregated monetary value representing a shareholder's
investments. In detail, it is defined for a
shareholder $i$ as
\begin{equation}
\label{eq:pfv}
p_i^{\text{dir}} = \sum_{j \in \Gamma(i)} W_{ij} v_j,
\end{equation}
where $\Gamma(i)$ is the set of indices of the neighbors of $i$,
denoting all the companies in the portfolio. In the presence of a
network, the notion of the indirect portfolio naturally arises
\cite{vitali2011network,glattfelder2019architecture}. This is the
value found in the portfolio of portfolios. Specifically, the indirect
portfolio value is found by traversing all the
indirect paths reachable downstream from $i$
\begin{equation}
\label{eq:phat}
\begin{split}
p_i^{\text{ind}} = &\sum_{j \in \Gamma(i)} \sum_{k \in \Gamma(j)}  W_{ij} W_{jk} v_k + \cdots +\\
&\sum_{j_1 \in \Gamma(i)}\sum_{j_2 \in \Gamma(j_1)}   \cdots \sum_{j_{m-1} \in \Gamma(j_{m})} 
 W_{i j_1} W_{j_1 j_2} \cdots  W_{j_{m-1} j_m} v_{j_m} + \cdots.
\end{split}
\end{equation}
As a result, one can assign the sum of the direct and indirect
portfolio values to each shareholder, retrieving the  total
  portfolio value
\begin{equation}
\label{eq:dipf}
p_i^{\text{tot}} = p_i^{\text{dir}} + p_i^{\text{ind}}.
\end{equation}
In matrix notation, this can be re-expressed as
\begin{equation}
\label{eq:adpow}
p_i^{\text{tot}}  = \sum_{n=1}^{\infty} W^n v.
\end{equation}
By virtue of Eq. (\ref{eq:wtilddef}), it is found that $\chi =
p_i^{\text{tot}}$.

In generic terms, Eq.  (\ref{eq:rec}) can be interpreted in the
context of a system in which a resource (e.g., energy or mass) is
flowing along the directed and weighted links of the network. In this
picture, the intrinsic property value $v_i$ associated with the nodes
represents the quantity of the resource produced by them. Now $\chi_i$
measures the inflow of this resource which accumulates in node $i$
from all the nodes downstream
\cite{glattfelder2009backbone,vitali2011network}.

In the following, this eigenvector centrality variant $\chi$ is called
the {\it access centrality}, as it computes how much a node can access
the intrinsic properties of all other nodes reachable downstream via
the direct and indirect weighted links.

\subsection{The Problems}

It was realized that the access centrality $\chi$ suffers from
undesirable issues.  When the number of cycles in the network is
large, for example in the core of the bow-tie, the measure computes
overestimated results, due to the nodes' intrinsic property value
flowing many times through the cycles. A remedy was proposed, where,
in essence, links are removed in the computation
\cite{baldone1998ownership,rohwer2005tda}. In detail, Eq.
(\ref{eq:matrixwtilddef}) is adapted as follows:
\begin{equation}
\label{eq:intownfix}
\widehat W_{ij} = W_{ij} + \sum_{k \neq i} \widehat W_{ik}W_{kj}.
\end{equation}
This is equivalent to the introduction of a correction matrix
\cite{vitali2011network,glattfelder2013decoding}
\begin{equation}
\label{eq:D}
\mathcal{D} = \textrm{diag}  \left( (\mathbb{1}-W)^{-1} \right)^{-1}.
\end{equation}
Recall that $\textrm{diag} (A)$ is defined as the matrix of the
diagonal elements of the matrix $A$. The components of $\mathcal{D}$
are
\begin{subequations}
 \label{eq:Dscalarall}
\begin{gather}
\mathcal{D}_{kk} = \frac{1}{(\mathbb{1}-W)^{-1}_{kk}},  \label{eq:Dscalara} \\
\mathcal{D}_{ij} = 0,  \quad i \neq j. \phantom{L}\label{eq:Dscalar}
\end{gather}
\end{subequations}
The corrected centrality matrix equation is
\begin{equation}
\label{eq:whatmat}
\widehat W =  \mathcal{D} \widetilde W,
\end{equation}
and the new {\it corrected centrality} emerging from these manipulations
is
\begin{equation}
\label{eq:chihatmat}
\widehat \chi = \widehat W v.
\end{equation}
As a result, in networks with cycles, by construction, $\widehat
\chi_i \leq \chi_i$, otherwise $\widehat \chi_i = \chi_i$.

While this approach remedies the problem of overestimating the
centrality due to cycles, it introduces another problem
\cite{vitali2011network,glattfelder2013decoding,glattfelder2019architecture}.
Cutting links in the network tames the cycles but makes root nodes
dominant.  A single root node in a bow-tie will have the highest
corrected centrality score, regardless of the level of
interconnectivity in the network. This is an undesirable effect, as
the topology becomes irrelevant.

In \cite{glattfelder2019architecture} the authors propose a centrality
measure addressing the above mentioned issues. In detail, an algorithm
is described which computes the centrality score for each node, called
the {\it influence index} $\xi$. For its computation, only the trails
in the network are traversed. These are unique paths where each node
is only visited once, thus terminating any further flow through
cycles.  In other words, for each iteration of the algorithm, the
calculation considers the jump from node to node along the directed
links until either a terminating leaf node is reached or a node that
was visited some steps earlier is detected (the result of a cycle).
In effect, the cycles in the network are cut.  This algorithmically
computed centrality represents a lower bound to the access centrality
$\xi_i \leq \chi_i$. It should be noted that the influence index
offers an algorithmic solution to the above mentioned problems.  In
this sense, it is a desirable centrality measure. The analytical
bow-tie centrality emulates its properties.

\subsection{An Example}

\begin{figure}[t]
\includegraphics[width=0.55\textwidth]{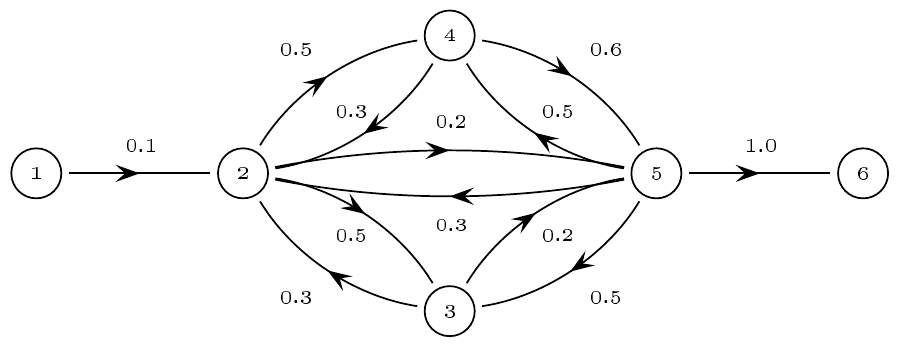}
\vspace*{8pt}
\caption{ Simple bow-tie network
  example with a high degree of interconnectedness of nodes  in the
  SCC. All nodes have unit value $v_i=1$. Reproduced from
  \cite{vitali2011network}.
}\label{fig:netexb}
\end{figure}

Fig. \ref{fig:netexb} presents the bow-tie example introduced
in \cite{vitali2011network,glattfelder2013decoding}. One finds for $v_i=1$
\begin{equation}
\chi = \left( 
\begin{array}{r}
    5\\
   49\\
   26\\
   48\\
   54\\
         0
\end{array} \right) , \ 
\widehat \chi = \left( 
\begin{array}{r}
    5.000\\
    4.900\\
    4.216\\
    4.571\\
    4.629\\
         0.000
\end{array} \right),
 \ 
\xi = \left( 
\begin{array}{r}
    0.360\\
    2.600\\
    1.400\\
    2.520\\
    3.050\\
         0.000
\end{array} \right).
\end{equation}
This simple example highlights the mentioned problems. Namely, the
overestimation of nodes in cycles affecting $\chi$ (i.e.,
$\chi_2$--$\chi_5$ having large values) and the dominance of root
nodes affecting $\widehat \chi$ (i.e., $\widehat \chi_1$ being larger
than $\widehat  \chi_2$ -- $\widehat  \chi_6$).  Moreover, it is confirmed
that $\xi$ indeed acts as the lower bound for the calculations.


\section{Deriving the Bow-Tie Centrality}

While the access centrality $\chi$ and the corrected centrality
$\widehat \chi$ are well-defined measures with clear interpretations,
in practice, as mentioned, they have drawbacks when applied to
weighted and directed networks with many cycles.  Another measure, the
influence index $\xi$, while remedying the problems can only be
defined algorithmically. In essence, what is missing in the literature
is an analytical expression (utilizing the adjacency matrix) which
allows a centrality score to be computed that is similar to the
influence index (and hence is also not plagued by the problems
detailed above) and which can be applied to bow-tie networks (with
intrinsic node properties).  In the following, such a new centrality
measure is derived.

Let the auxiliary matrix $V$ be defined as
\begin{equation}
\label{eq:vdef}
V =   (\mathbb{1}- W)^{-1}.
\end{equation}
This allows some equations to be re-expressed. Eq. (\ref{eq:wtilddef}) is now
\begin{equation}
\label{wq:wtildev}
\widetilde W = W V
\end{equation}
and  Eq. (\ref{eq:whatmat}) becomes
\begin{equation}
\label{wq:whatv}
\widehat W = \mathcal{D} \widetilde W = \mathcal{D} W V,
\end{equation}
highlighting how the correction matrix impacts the original centrality
matrix. However, the correction matrix $\mathcal{D}$ could be applied
at a different position in Eq. (\ref{wq:whatv}), unveiling yet
another corrected eigenvector centrality variant
\begin{equation}
\label{wq:wbarev}
\overline  W = W \mathcal{D} V.
\end{equation}
In essence, the non-commutative nature of matrix multiplication
$\mathcal{D} W \neq W \mathcal{D}$ results in two analytical
expressions. It should be noted that in mathematics and physics, 
non-commutative behavior is a source of rich
structure \cite{connes1994noncommutative,seiberg1999string,douglas2001noncommutative}.

For the new centrality measure, one finds
\begin{equation}
\label{wq:wbar}
\overline  W = W \mathcal{D}  (\mathbb{1}- W)^{-1} =  W \mathcal{D}  \left(W (\mathbb{1}- W)^{-1}+\mathbb{1} \right)
= W \mathcal{D}  (\widetilde W + \mathbb{1}) = W (\widehat W +  \mathcal{D}),
\end{equation}
by noting that from Eq. (\ref{eq:iminaassum}) 
\begin{equation}
\mathbb{1} + W + W^2 + \dots = W ( \mathbb{1} + W + W^2 + \dots) + \mathbb{1}.
\end{equation}

In a nutshell, the new centrality matrix is
\begin{equation}
\label{eq:Wbar}
\overline W =  W  W^\ast,
\end{equation}
with
\begin{equation}
W^\ast = \widehat{W} + \mathcal D, 
\end{equation}
or in scalar notation
\begin{equation}
W^{\ast}_{ij} = 
\begin{cases}
1,  &i=j, \\
\widehat{W}_{ij},  &i \neq j.
\end{cases}
\end{equation}
The final resulting centrality measure, called the {\it bow-tie
  centrality}, is 
\begin{equation}
\label{eq:zeta}
\zeta  = \overline W v.
\end{equation}
It is an analytical expression that does not overestimate the importance
of nodes in cycles and root nodes.

In the example seen in Fig. \ref{fig:netexb} one finds
\begin{equation}
\chi = \left( 
\begin{array}{r}
    5\\
   49\\
   26\\
   48\\
   54\\
         0
\end{array} \right) , \ 
\zeta = \left( 
\begin{array}{r}
    0.500\\
    5.465\\
    2.443 \\
    4.329\\
   7.023 \\
         0.000
\end{array} \right),
 \ 
\xi = \left( 
\begin{array}{r}
    0.360\\
    2.600\\
    1.400\\
    2.520\\
    3.050\\
    0.000
\end{array} \right).
\end{equation}

\begin{theorem}\label{thm1}
  The bow-tie centrality measure $\zeta$ is bounded by the access
  centrality $\chi$ and the influence index $\xi$:
\begin{equation}
\chi_i \geq \zeta_i \geq \xi_i .
\end{equation}
\end{theorem}

\begin{proof}
  Recall that the adjacency matrix of an ownership network is
  column-stochastic if all the ownership information is known. In
  general, $\sum_j W_{ij} \leq 1$ and $W_{ij} \in[0,1]$. The matrix $V$
  is defined in Eq. (\ref{eq:vdef}) and from Eq. (\ref{eq:iminaassum}),
  $V_{ij} \geq 0$.  

  The first inequality translates into $ \widetilde W_{ij} v_j \geq
  \overline W_{ij} v_j$. Componentwise
\begin{equation}
 \overline  W_{ij} = \sum_{k,l} W_{ik} \mathcal{D}_{kl} V_{lj}
\stackrel{Eq. (\ref{eq:Dscalar})}{=} \sum_k W_{ik} \mathcal{D}_{kk} V_{kj} \leq
 \sum_k W_{ik} V_{kj}  =  \widetilde W_{ij}.
\end{equation}
The inequality follows by observing that from Eq. (\ref{eq:Dscalara}),
$\mathcal{D}_{kk} = V_{kk}^{-1} \in ]0,1]$, because $V_{kk} = 1 + W_{kk} +
\dots \geq 1$, and thus $\mathcal{D}_{kk} V_{kj} \leq  V_{kj}$.

The second inequality arises by construction. In the absence of
cycles, all the centrality measures are identical. In the presence of
cycles, the influence index algorithm only traverses the trails in the
network.  These are unique paths where each node is only visited once,
thus terminating any further flow through cycles. In other words,
cycles are never actually fully traversed and the computation stops
one node before completion.  This is the minimum possible contribution
from cycles. The bow-tie centrality is defined using powers of $W$,
resulting in paths of various lengths being considered. In essence,
contributions from cycles are incorporated in the computation,
exceeding the bare minimum arising from the trails. It should be noted
that if a centrality measure yields lower values than the influence
index this means that information from the cycles has been lost in the
computation.
\end{proof}

\section{Empirical Application}

A prototypical bow-tie network can be found in the global ownership
network \cite{glattfelder2019architecture}. Ownership networks are
comprised of economic actors which are connected via ownership
relations (i.e., by holding a percentage of a corporations' equity).
Shareholders can be other firms, natural persons, families,
foundations, research institutes, public authorities, states, and
government agencies. Ownership networks are directed and weighted, and
some of the corporations have an intrinsic node property, coming in
the guise of an economic value (e.g., the firm's operating revenue in
USD). Typical for ownership networks is the emergence of a tiny but
highly interconnected core (SCC) of influential shareholders. Hence,
this is an ideal real-world use-case for the bow-tie centrality
measure.

The global ownership information is taken from Bureau van Dijk's Orbis
database\footnote{See
  \url{http://www.bvdinfo.com/en-gb/our-products/company-information/international-products/orbis}.}.
In \cite{glattfelder2019architecture}, six yearly network snapshots
were constructed and analyzed from this. Here, we focus on the largest
connected component of 2012. In detail, we analyzes the IN, the SCC,
and the OUT, omitting the TT. From a set of 35,839,090 nodes and
27,307,642 links, the largest connected component was identified as
being comprised of 5,933,836 nodes. In the following, a subnetwork of
the 2012 global ownership network is analyzed. By considering all of
the IN and SCC nodes, plus all OUT nodes with an operating revenue
larger or equal to USD 100,000,000 (i.e., $v_i^{\text{out}} \geq 100$
million), a reduced ownership network is retrieved. It contains 64,266
nodes and 540,405 links. In Table \ref{tab:nwdyn} the bow-tie
component sizes are shown.  There are 52,001 nodes with an operating
revenue $v_i >0$ contained in the reduced network, totaling USD
84,287,655,740,000.
This represents 66.30\% of the total global operating revenue of the
entire global ownership network, which is approximately USD 127
trillion \cite{glattfelder2019architecture}.

\begin{table}[t!]
\center
\begin{tabular}{ l | r}
  IN  &13,374\\
  SCC& 2,554\\
  OUT&48,338\\
    \hline   
Total &  64,266\\
\end{tabular} 
\caption{
Bow-tie components of the reduced global ownership network of 2012. The number of nodes in
    the various components of the network are shown, totaling 64,266 nodes. The different
    centrality measures are applied to this empirical network.
}
  \label{tab:nwdyn}
\end{table}

\subsection{Comparing the Rankings}

For the empirical analysis, all the discussed centrality measures are
computed for this network. Namely
\begin{enumerate}
\item the access centrality  $\chi = (\mathbb{1}- W)^{-1} W v = \widetilde W v$
\item the corrected centrality $\widehat \chi =  \mathcal{D} \widetilde W v = \widehat W v$
\item the bow-tie centrality $\zeta = W (\widehat{W} + \mathcal D) v = \overline W v$
\item the algorithmic influence index $\xi$.
\end{enumerate}
As a result, there are six possible comparisons of the rankings. By
utilizing the Jaccard index, a statistic used for assessing the
similarity of sets \cite{jaccard1912distribution}, such a comparison
can be quantified. The index is defined as
\begin{equation}
\mathcal{J}(A,B) = {{|A \cap B|}\over{|A \cup B|}},
\end{equation}
for two sets $A$ and $B$.  A value of one indicates a total overlap of
the sets, while zero denotes no similarity. 

The truncated Jaccard index $\mathcal{J}_{n}$ is employed to uncover
the similarity between the $n$ highest ranked nodes of two centrality
measures \cite{glattfelder2019architecture}. For instance,
$\mathcal{J}_{n} (\zeta, \chi)$ compares the bow-tie centrality with
the access centrality for the firs $n$ nodes ranked by centrality.
$\mathcal{J}_{n}$ can then be plotted for all centrality comparisons
of all lengths.


%
\begin{figure}[t]
\includegraphics[width=0.5\textwidth]{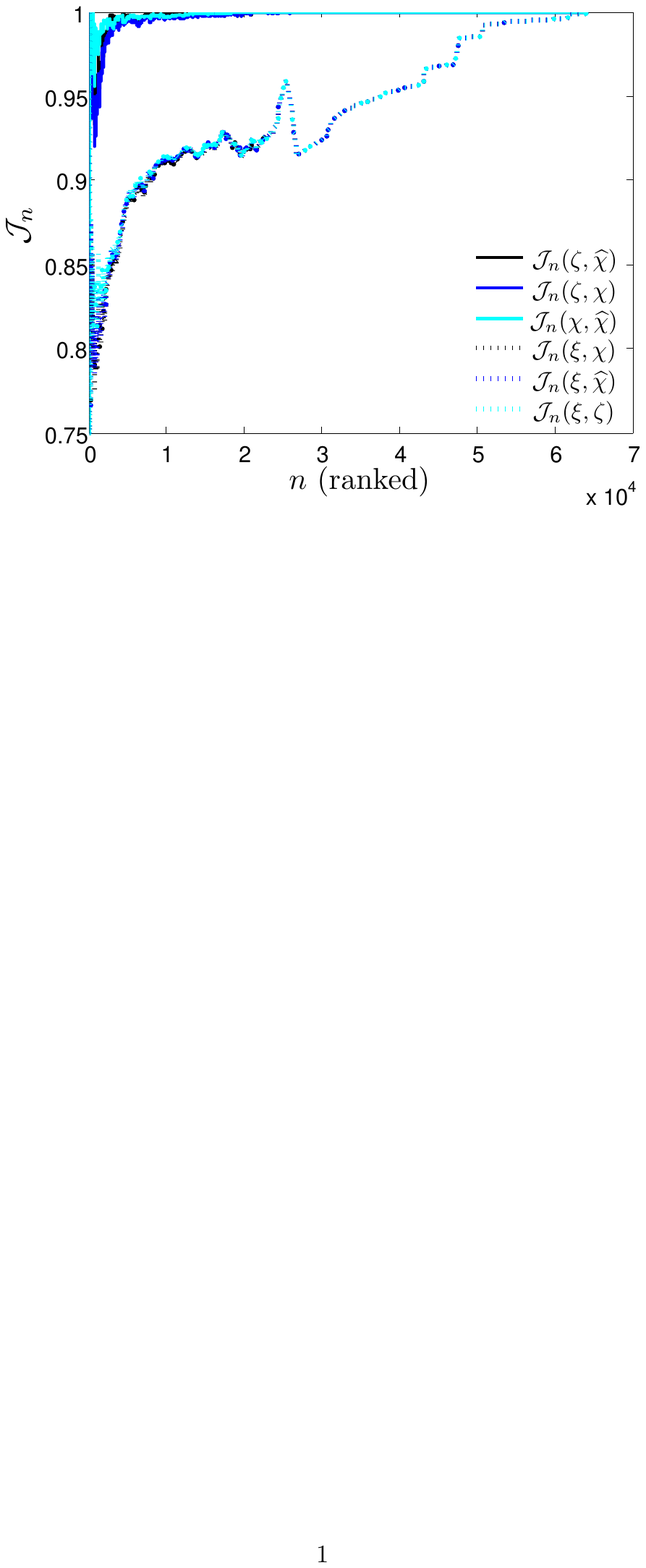}
\vspace*{8pt}
\caption{Jaccard indices comparing the centrality rankings. The
  similarity of the four network centrality measures is visualized by
  computing the Jaccard index $\mathcal{J}_{n}$ for all ranked sets of
  length $n=1,\dots, 64,266$. The $x$-axis shows the number of
  centrality-ranked nodes, where $n=1$ represents the node with the
  highest centrality.  See discussion in text.}
\label{fig:jacc}
\end{figure}

Fig. \ref{fig:jacc} shows the results of this computation. The three
analytical centrality measures $\chi, \widehat \chi$, and $\zeta$
display a high similarity among themselves. In contrast, the
algorithmic centrality $\xi$ shows a pronounced dissimilarity with
respect to the analytical measures. The two centrality variants based
on $\chi$, utilizing the correction matrix $\mathcal{D}$, show the
highest similarity.  Symbolically, $\zeta \sim \widehat \chi$. This
implies that the way the bow-tie centrality adjusts for cycles yields
similar results as the corrected centrality. Recalling that both
$\chi$ and $\widehat \chi$ are affected by topology-dependent issues,
the bow-tie centrality $\zeta$ emerges as a superior analytical
centrality measure, while still accurately reflecting the behavior of
the access centrality $\chi$.

\begin{figure}[t]
\centering
\includegraphics[width=0.5\textwidth]{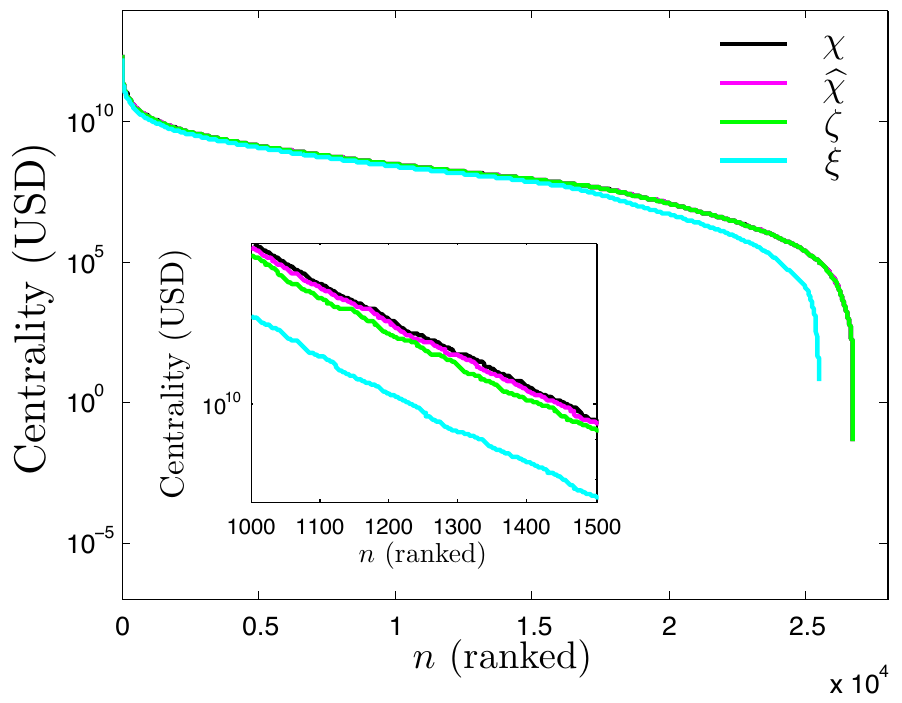}
\vspace*{8pt}
\caption{Ranked centrality value semi-log plots. For all four centrality
  measures, their monetary value is shown for each node. The inset
  shows the details for 500 nodes ranked between 10,000 and 15,000.}
\label{fig:cent}
\end{figure}

By further comparing the bow-tie centrality to the influence index,
the following is revealed. Of the 64,266 nodes in the reduced
ownership network, 26,758 have a non-zero analytical centrality.
Symbolically, $|\{ \zeta_i > 0 \}| = |\{ \chi_i > 0\}| = |\{ \widehat
\chi_i > 0\}| = 26,758 $. For the algorithmic influence index, the
number is the following $|\{ \xi_i > 0 \}| = 25,512$. From Theorem
\ref{thm1} it is known that the influence index is a lower bound for
the bow-tie centrality, i.e., $\zeta_i \geq\xi_i$. In Fig.
\ref{fig:cent} the monetary values of each node for each centrality
variant is shown in ranked plots.  Indeed, in the network at hand, the
algorithmic centrality is below the analytical ones, which all have
similar values.

\subsection{Top-Ten Rankings}

In \cite{glattfelder2019architecture} the entire global ownership
network was analyzed. In other words, the centrality scores for
35,839,090 nodes was algorithmically evaluated (i.e., $\xi_i$).  From
a computational perspective, calculating 
the analytical centrality measures for such a matrices can be
challenging. 

However, by focusing on the core structures of the network, a smaller
representation can be found. The reduced ownership network analyzed
here, comprised of 64,266 nodes, represents 66.30\% of the total
global operating revenue. A key question to be answered is how
representative this subnetwork is? Moreover, as a different centrality
measure was used, are the two resulting rankings comparable?  In other
words, how similar is the influence index, computed for the entire
network, to the bow-tie centrality, calculated for the reduced
network? 

In Table \ref{t:comprank} the two centralities are compared for the
top-ten list. While the individual rankings are expected to change,
the group of actors is very similar. Both centrality measures
crucially identify a similar set of most relevant nodes in the
network.  Furthermore, although there is not as much value in the
reduced network, the bow-tie centrality is a less conservative
estimation resulting in higher centrality values.  

In summary, this result can be seen as evidence that it suffices to
focus on the relevant part of the network for quantitative analysis.
In this case it was a combination of topology (all the nodes in the IN
and SCC) and value (nodes in the OUT with a minimum of 100 million
operating revenue in USD).  For rapid testing of ideas and
approximating indicators, this approach can be invaluable.  The
analytical bow-tie centrality is straightforward to apply while the
implementation of the algorithmic influence index required the
utilization of a graph database.

\begin{table}[t]
\center
\footnotesize
\begin{tabular}{ lr |  lr } 
\multicolumn{2}{ c }{\bf Influence Index} (entire network)& \multicolumn{2}{ c }{\bf Bow-Tie Centrality} (reduced network)\\
Name                           &            $\xi$  (t USD)  & Name                         &           $\zeta$   (t USD)   \\ \hline  \\ [-10.5px]
   BLACKROCK INC              & 2.177           &BLACKROCK INC&  2.421 \\
 VANGUARD GROUP INC               & 1.314       & VANGUARD GROUP INC  &   1.516       \\
  GOVERNMENT OF NORWAY          & 1.220        &STATE STREET CORP  &   1.359     \\
  SASAC                           & 1.210         &  GOVERNMENT OF NORWAY &  1.299\\
 CAPITAL GROUP    COMPANIES     & 1.201       & CAPITAL GROUP COMPANIES &    1.246    \\
 STATE STREET CORP               &1.190          &BPCE SA&  1.155    \\
 GOVERNMENT OF FRANCE            &0.982       & FMR LLC  &  1.096    \\
 FMR LLC                        &0.955             & BARCLAYS PLC &   0.796\\
 BARCLAYS PLC               &0.675           &JP  MORGAN CHASE  \& CO& 0.715     \\
 ROYAL DUTCH SHELL PLC              &0.619   &T. ROWE PRICE GROUP INC&  0.643   
\end{tabular}
\caption{Comparing the top-ten rankings. The exhaustive results taken from  \citep{glattfelder2019architecture} are seen on the left-hand side. The right-hand side shows the results from the reduced network utilizing the bow-tie centrality from Eq. (\ref{eq:zeta}). The values are in trillion USD.}
\label{t:comprank}
\end{table}

\section{Conclusion}

Which are the most important nodes in a network? This question has a
long history in network science and there are different ways of
approaching it. For instance, algorithms can be developed which
traverse the network and compute the centrality scores of the nodes.
While such an approach requires a computational framework, it can be
applied to very large networks. A more classical approach is to
utilize equations. In this analytical context many centrality measures
have been proposed. However, for a relevant class of directed and
weighted real-world networks, characterized by a bow-tie topology,
these centralities suffer from drawbacks.  Specifically, the emergence
of cycles represents a formidable challenge.

We introduce a novel centrality measure ideally applied to networks
displaying a bow-tie topology. The quantity represents a final
iteration in a stream of research originating from the study of
ownership networks
\cite{brioschi1989risk,briosch1995equity,baldone.ea98,rohwer2005tda,glattfelder2009backbone,vitali2011network,glattfelder2013decoding,glattfelder2019architecture}.
The new centrality measure can be applied in general to weighted and
directed complex networks, where the nodes carry an intrinsic
non-topological degree of freedom. The ideal domain of application are
such networks, where the bow-tie contains important nodes. Here, older
centrality measures overestimate the relevance of such nodes (in the
SCC and IN bow-tie components) and blur important features. Indeed,
the bow-tie centrality yields a precise score for every node in the
network, regardless of its location in the bow-tie.

The bow-tie centrality can be clearly motivated analytically and does
not possesses the undesirable features plaguing older variants (such
as the access and corrected centralities). Comparing these different
centrality measure to each other reveals that the novel bow-tie
centrality achieves this while still capturing the desired features of
ownership-inspired eigenvector centralities.  In essence, this
centrality represents the analytical counterpart to an algorithmic
implementation used to decode empirical ownership networks, namely the
influence index \cite{glattfelder2019architecture} (or an older, more
cumbersome variant found in \cite{vitali2011network}).

Finally, it was demonstrated that large networks can be reduced to
smaller subsets which, when analyzed, show very similar properties as
the whole. In essence, the characteristic features of a network are
encoded in the subnetwork of important nodes, making the detection of
this backbone crucial.  To this aim, better centrality measures are
important.

\section*{Acknowledgments}
I would like to thank Stefano Battiston and Borut Sluban for their
support.


\bibliographystyle{abbrvnat}
\bibliography{ref}

\begin{thebibliography}{33}
\providecommand{\natexlab}[1]{#1}
\providecommand{\url}[1]{\texttt{#1}}
\expandafter\ifx\csname urlstyle\endcsname\relax
  \providecommand{\doi}[1]{doi: #1}\else
  \providecommand{\doi}{doi: \begingroup \urlstyle{rm}\Url}\fi

\bibitem[Baldone et~al.(1998{\natexlab{a}})Baldone, Brioschi, and
  Paleari]{baldone.ea98}
S.~Baldone, F.~Brioschi, and S.~Paleari.
\newblock Ownership measures among firms connected by cross-shareholdings and a
  further analogy with input-output theory.
\newblock \emph{4th JAFEE International Conference on Investment and
  Derivatives}, 1998{\natexlab{a}}.

\bibitem[Baldone et~al.(1998{\natexlab{b}})Baldone, Brioschi, and
  Paleari]{baldone1998ownership}
S.~Baldone, F.~Brioschi, and S.~Paleari.
\newblock {Ownership Measures Among Firms Connected by Cross-Shareholdings and
  a Further Analogy with Input-Output Theory}.
\newblock \emph{4th JAFEE International Conference on Investment and
  Derivatives}, 1998{\natexlab{b}}.

\bibitem[Benevenuto et~al.(2009)Benevenuto, Rodrigues, Almeida, Almeida, and
  Ross]{benevenuto2009video}
F.~Benevenuto, T.~Rodrigues, V.~Almeida, J.~Almeida, and K.~Ross.
\newblock Video interactions in online video social networks.
\newblock \emph{ACM Transactions on Multimedia Computing, Communications, and
  Applications (TOMCCAP)}, 5\penalty0 (4):\penalty0 30, 2009.

\bibitem[Bonacich(1972)]{bonacich1972factoring}
P.~Bonacich.
\newblock Factoring and weighting approaches to status scores and clique
  identification.
\newblock \emph{Journal of Mathematical Sociology}, 2\penalty0 (1):\penalty0
  113--120, 1972.

\bibitem[Bonacich(1987)]{bonacich1987power}
P.~Bonacich.
\newblock Power and centrality: {A} family of measures.
\newblock \emph{The American Journal of Sociology}, 92\penalty0 (5):\penalty0
  1170--1182, 1987.

\bibitem[Bonacich and Lloyd(2001)]{bonacich2001eigenvector}
P.~Bonacich and P.~Lloyd.
\newblock {Eigenvector-Like Measures of Centrality for Asymmetric Relations}.
\newblock \emph{Social Networks}, 23\penalty0 (3):\penalty0 191--201, 2001.

\bibitem[Borgatti and Everett(2006)]{borgatti2006graph}
S.~Borgatti and R.~Everett.
\newblock {A Graph-Theoretic Perspective on Centrality}.
\newblock \emph{Social Networks}, 28\penalty0 (4):\penalty0 466--484, 2006.

\bibitem[Brin and Page(1998)]{brin1998anatomy}
S.~Brin and L.~Page.
\newblock {The Anatomy of a Large-Scale Hypertextual Web Search Engine}.
\newblock \emph{Computer Networks and ISDN Systems}, 30\penalty0
  (1-7):\penalty0 107--117, 1998.

\bibitem[Brioschi and Paleari(1995)]{briosch1995equity}
F.~Brioschi and S.~Paleari.
\newblock {How Much Equity Capital Did the Tokyo Stock Exchange Really Raise}.
\newblock \emph{Financial Engineering and the Japanese Markets}, 2\penalty0
  (3):\penalty0 233--258, 1995.

\bibitem[Brioschi et~al.(1989)Brioschi, Buzzacchi, and
  Colombo]{brioschi1989risk}
F.~Brioschi, L.~Buzzacchi, and M.~Colombo.
\newblock {Risk Capital Financing and the Separation of Ownership and Control
  in Business Groups}.
\newblock \emph{Journal of Banking and Finance}, 13\penalty0 (1):\penalty0 747
  -- 772, 1989.

\bibitem[Broder et~al.(2000)Broder, Kumar, Maghoul, Raghavan, Rajagopalan,
  Stata, Tomkins, and Wiener]{broder00}
A.~Broder, R.~Kumar, F.~Maghoul, P.~Raghavan, S.~Rajagopalan, S.~Stata,
  A.~Tomkins, and J.~Wiener.
\newblock {Graph Structure in the Web}.
\newblock \emph{Computer Networks}, 33:\penalty0 309, 2000.

\bibitem[Connes(1994)]{connes1994noncommutative}
A.~Connes.
\newblock \emph{Noncommutative geometry}.
\newblock Academic Press, San Diego, 1994.

\bibitem[Donato et~al.(2008)Donato, Leonardi, Millozzi, and
  Tsaparas]{donato2008mining}
D.~Donato, S.~Leonardi, S.~Millozzi, and P.~Tsaparas.
\newblock {Mining the Inner Structure of the Web Graph}.
\newblock \emph{Journal of Physics A: Mathematical and Theoretical},
  41\penalty0 (22):\penalty0 224017, 2008.

\bibitem[Douglas and Nekrasov(2001)]{douglas2001noncommutative}
M.~R. Douglas and N.~A. Nekrasov.
\newblock Noncommutative field theory.
\newblock \emph{Reviews of Modern Physics}, 73\penalty0 (4):\penalty0 977,
  2001.

\bibitem[Freeman(1978)]{freeman1978centrality}
L.~Freeman.
\newblock {Centrality in Social Networks Conceptual Clarification}.
\newblock \emph{Social Networks}, 1:\penalty0 215--239, 1978.

\bibitem[Freeman(2008)]{freeman2008going}
L.~Freeman.
\newblock {Going the Wrong Way on a One-Way Street: Centrality in Physics and
  Biology}.
\newblock \emph{Journal of Social Structure}, 9\penalty0 (2), 2008.

\bibitem[Fujiwara and Aoyama(2008)]{fujiwara2008large}
Y.~Fujiwara and H.~Aoyama.
\newblock Large-scale structure of a nation-wide production network.
\newblock \emph{The European Physical Journal B-Condensed Matter and Complex
  Systems}, pages 1--16, 2008.

\bibitem[Glattfelder(2013)]{glattfelder2013decoding}
J.~B. Glattfelder.
\newblock \emph{{Decoding Complexity: Uncovering Patterns in Economic
  Networks}}.
\newblock Springer, Heidelberg, 2013.

\bibitem[Glattfelder and Battiston(2009)]{glattfelder2009backbone}
J.~B. Glattfelder and S.~Battiston.
\newblock Backbone of complex networks of corporations: The flow of control.
\newblock \emph{Physical Review E}, 80\penalty0 (3):\penalty0 036104, 2009.

\bibitem[Glattfelder and Battiston(2019)]{glattfelder2019architecture}
J.~B. Glattfelder and S.~Battiston.
\newblock The architecture of power: Patterns of disruption and stability in
  the global ownership network.
\newblock SSRN, 2019.

\bibitem[Hubbell(1965)]{hubbell1965input}
C.~Hubbell.
\newblock An input-output approach to clique identification.
\newblock \emph{Sociometry}, 28\penalty0 (4):\penalty0 377--399, 1965.

\bibitem[Jaccard(1912)]{jaccard1912distribution}
P.~Jaccard.
\newblock The distribution of the flora in the alpine zone.
\newblock \emph{New Phytologist}, 11\penalty0 (2):\penalty0 37--50, 1912.

\bibitem[Katz(1953)]{katz1953new}
L.~Katz.
\newblock A new status index derived from sociometric analysis.
\newblock \emph{Psychometrika}, 18\penalty0 (1):\penalty0 39--43, 1953.

\bibitem[Leontief(1966)]{leontief1986input}
W.~Leontief.
\newblock \emph{{Input-Output Economics}}.
\newblock Oxford University Press, Oxford, 1966.

\bibitem[L{\"u} et~al.(2016)L{\"u}, Chen, Ren, Zhang, Zhang, and
  Zhou]{lu2016vital}
L.~L{\"u}, D.~Chen, X.-L. Ren, Q.-M. Zhang, Y.-C. Zhang, and T.~Zhou.
\newblock Vital nodes identification in complex networks.
\newblock \emph{Physics Reports}, 650:\penalty0 1--63, 2016.

\bibitem[Martin et~al.(2014)Martin, Zhang, and Newman]{martin2014localization}
T.~Martin, X.~Zhang, and M.~E. Newman.
\newblock Localization and centrality in networks.
\newblock \emph{Physical review E}, 90\penalty0 (5):\penalty0 052808, 2014.

\bibitem[Newman et~al.(2006)Newman, Barab{\'a}si, and
  Watts]{newman2006structure}
M.~Newman, A.~Barab{\'a}si, and D.~Watts.
\newblock \emph{{The Structure and Dynamics of Networks}}.
\newblock Princeton University Press, Princeton, 2006.

\bibitem[Page et~al.(1999)Page, Brin, Motwani, and Winograd]{page1999pagerank}
L.~Page, S.~Brin, R.~Motwani, and T.~Winograd.
\newblock The pagerank citation ranking: Bringing order to the web.
\newblock Technical report, Stanford InfoLab, 1999.

\bibitem[Rohwer and P\"otter(2005)]{rohwer2005tda}
G.~Rohwer and U.~P\"otter.
\newblock \emph{{Transition Data Analysis User's Manual}}.
\newblock Ruhr-University Bochum, 2005.

\bibitem[Schweitzer et~al.(2009)Schweitzer, Fagiolo, Sornette, Vega-Redondo,
  Vespignani, and White]{schweitzer2009economicsci}
F.~Schweitzer, G.~Fagiolo, D.~Sornette, F.~Vega-Redondo, A.~Vespignani, and
  D.~White.
\newblock {Economic Networks: The New Challenges}.
\newblock \emph{Science}, 325\penalty0 (5939):\penalty0 422, 2009.

\bibitem[Seiberg and Witten(1999)]{seiberg1999string}
N.~Seiberg and E.~Witten.
\newblock String theory and noncommutative geometry.
\newblock \emph{Journal of High Energy Physics}, 1999\penalty0 (09):\penalty0
  032, 1999.

\bibitem[Vitali et~al.(2011)Vitali, Glattfelder, and
  Battiston]{vitali2011network}
S.~Vitali, J.~B. Glattfelder, and S.~Battiston.
\newblock The network of global corporate control.
\newblock \emph{PLoS one}, 6\penalty0 (10):\penalty0 e25995, 2011.

\bibitem[Zhang et~al.(2007)Zhang, Ackerman, and Adamic]{zhang2007expertise}
J.~Zhang, M.~Ackerman, and L.~Adamic.
\newblock Expertise networks in online communities: structure and algorithms.
\newblock In \emph{Proceedings of the 16th international conference on World
  Wide Web}, pages 221--230. ACM, 2007.

\end{thebibliography}

\end{document}